\documentclass[11pt]{amsart}

\usepackage[T1]{fontenc}
\usepackage[utf8]{inputenc}
\usepackage[english]{babel} 
\usepackage{amsmath,amssymb,amsfonts,amsthm}
\usepackage{mathrsfs, esint, dsfont}
\usepackage{moreverb,rotating,graphics,float}
\usepackage{caption, subcaption}
\usepackage[colorlinks, linkcolor={blue}, citecolor={red}]{hyperref}
\usepackage{framed,fancybox}
\usepackage{enumerate}
\usepackage{slashed}
\usepackage{cancel}

\numberwithin{equation}{section} 
\newtheorem{theorem}{Theorem}[section]
\newtheorem{proposition}[theorem]{Proposition}
\newtheorem{remark}[theorem]{Remark}
\newtheorem{lemma}[theorem]{Lemma}

\newcommand\1{{\mathds 1}}

\newcommand{\RR}{\mathbb{R}}
\newcommand{\ZZ}{\mathbb{Z}}

\newcommand{\ri}{\mathrm{i}}
\newcommand{\cS}{\mathscr{S}}
\newcommand{\cD}{\mathcal{D}}
\newcommand{\cE}{\mathcal{E}}

\newcommand{\cC}{\mathcal{C}}
\newcommand{\cR}{\mathcal{R}}

\newcommand{\R}{\mathbb{R}}
\newcommand{\N}{\mathbb{N}}

\newcommand\bS{{\bold S}}

\def\cC{{\mathcal C}}
\def\cD{{\mathcal D}}
\def\cE{{\mathcal E}}
\def\cF{{\mathcal F}}
\def\cG{{\mathcal G}}

\def\cI{{\mathcal I}}

\def\cK{{\mathcal K}}

\def\cM{{\mathcal M}}

\def\cR{{\mathcal R}}
\def\cS{{\mathcal S}}

\def\cU{{\mathcal U}}
\def\rV{{\mathcal V}}

\def\fS{{\mathfrak S}}

\def\rd{{\mathrm{d}}}
\def\re{{\mathrm{e}}}
\def\ri{{\mathrm{i}}}

\def\rW{{\mathrm{W}}}
\def\rV{{\mathrm{V}}}

\newcommand{\TF}{{\rm TF}}

\newcommand{\per}{{\rm per}}
\newcommand{\loc}{{\rm loc}}
\newcommand{\Tr}{{\rm Tr}}

\newcommand{\VTr}{\underline{\rm Tr}}

\newcommand{\kin}{{\rm kin}}

\newcommand{\set}[1]{\left\{ #1\right\}}
\newcommand{\bra}[1]{\left( #1\right)}
\newcommand{\av}[1]{\left| #1\right|}
\newcommand{\com}[1]{\left[ #1\right]}

\newcommand\ii{{\infty}}

\numberwithin{equation}{section}

\title[DFT for 1D homogeneous materials]{On Density Functional Theory models for one-dimensional homogeneous materials}

\author[B. Bensiali]{Bouchra Bensiali}
\address[B. Bensiali]{Ecole Centrale Casablanca, Bouskoura, Ville Verte, P.B. 27182, Morocco}
\email{bouchra.bensiali@centrale-casablanca.ma}

\author[S. Lahbabi]{Salma Lahbabi}
\address[S. Lahbabi]{EMAMI, LRI, ENSEM, UHIIC, 7 Route d’El Jadida, B.P. 8118 Oasis, Casablanca; 
    UM6P-College of Computing, Lot 660, Hay Moulay Rachid Ben Guerir, 43150,
    Morocco}
\email{s.lahbabi@ensem.ac.ma}

\author[A. Maichine]{Abdallah Maichine}
\address[A. Maichine]{LAMA, Department of Mathematics, Faculty of Sciences, Mohammed V University in Rabat, 4 Av. Ibn Battouta, 1014, Rabat, Morocco}
\email{a.maichine@um5r.ac.ma}

\author[O. Mirinioui]{Othmane Mirinioui}
\address[O. Mirinioui]{Ecole Centrale Casablanca, Bouskoura, Ville Verte, P.B. 27182, Morocco}
\email{Othmane.Mirinioui@centrale-casablanca.ma}
\date{\today}

\begin{document}

\begin{abstract}
    This paper studies DFT models for homogeneous 1D materials in the 3D space. It follows our previous work about DFT models for homogeneous 2D materials in 3D. 
    We show how to reduce the problem from a 3D energy functional to a \emph{2D} energy functional. The kinetic energy is treated as in~\cite{GLM21,GLM23} by diagonilizing admissible states, and writing the kinetic energy as the infimum of a modified kinetic energy functional on reduced states. Besides, we treat here the Hartree interaction term in 2D, and show how to properly define the mean-field potential, through Riesz potential. We then show the well posdness of the reduced model and present some numerical illustrations. 
\end{abstract}

\maketitle

\tableofcontents


\section{Introduction}
In recent years, low-dimensional materials, such as 2D sheets, 1D nanowires and quantum dots 
 have become of great importance in condensed matter physics due to their distinctive features, which are relevant in many applications~\cite{VdW,review-low,Review-graphene,review-phosphorene, review-MOS2, 1D,0D}.  In order to investigate  their electronic properties, such as electron's mobility and band gap, Density Functional Theory (DFT) calculations are often used 
 ~\cite{calc-low, review-calcul-properties,review-calcul-properties-2, 1D-comp, method-2D,modeling-nano}.  Our objective is to prove and study DFT models for low-dimensional systems in order to allow low computational cost. In~\cite{BlancLeBr00}, the Thoma-Fermi-von Weisäcker model for 1D and 2D materials is derived by means of thermodynamic limit procedure. The works~\cite{GLM21,GLM23} have initiated the study of DFT models for low dimensional materials that are homogeneous, and we were particularly interested in 2D materials embedded in the 3D space. In this paper, we focus on one--dimensional materials that have an atomic length in two dimensions of the space, and a crystalline structure in the remaining  dimension. 

\medskip 
In~\cite{GLM21}, 2D materials are considered with a nuclear charge  distribution of the form
 $$\mu_{2D}(x_1,x_2,x_3)=\mu_{2D}(x_3)$$ 
 that depends only on the orthogonal direction $x_3$. The electrons are described by one body density matrices $\gamma$ that satisfy Pauli principle 
  $0\leq \gamma\leq 1$. 
 As the problem is convex, $\gamma$ share the same symmetry as $\mu$, 
 \begin{equation}\label{eq:sym2D}
  U_R \gamma=\gamma U_R,\quad \forall R\in \RR^2,
 \end{equation}
where $U_Rf(x)=f(x_1-R_1,x_2-R_2,x_3)$ is the translation operator along the first two variables. 
 Note that the density associated to a $\gamma$ satisfying~\eqref{eq:sym2D} depends only on the third variable $\rho_\gamma(x)=\rho_\gamma(x_3)$ and the  trace per unit surface of $\gamma$ can be defined by $\underline{\Tr}_{2D}(\gamma)=\int_\RR \rho_\gamma$. In the reduced Hartree-Fock (rHF) model, the energy of the system is of the form
$$
\underline{\cE}_{2D}(\gamma)=\frac12\underline{\Tr}_{2D}\bra{-\Delta_3 \gamma}+ \cD_1(\rho_\gamma-\mu), 
$$
where  $\Delta_d$ denotes the  the Laplacian operator on $L^2(\R^d)$, and $\cD_1$ is the 1D Hartree interaction. 
One main result of~\cite{GLM21} is a thorough discussion of this term and the definition of the mean-field potential (see~\cite[Proposition 3.3]{GLM21}). We prove that this problem is equivalent to a one-dimensional problem set on  operators $G$ acting on $L^2(\RR)$, which are positive $G \geq 0$, but which no longer need to satisfy 
Pauli principle. Instead, an additional term appears in the energy: 
the three-dimensional kinetic energy $\frac12\underline{\Tr}_{2D}(-\Delta_3\gamma)$
is replaced by a one-dimensional kinetic energy
of the form
$$
\frac12\Tr(-\Delta_1 G)+\pi \Tr(G^2). 
$$
Although we restrict ourselves to the reduced Hartree-Fock model for simplicity, similar derivations
can be performed for general Kohn--Sham models and one can consider models of the form 
$$
\inf\set{\frac12\Tr(-\Delta_1 G)+ \pi\Tr(G^2)+ \cD_1(\rho_G-\mu) + E^{\rm xc}(\rho_G)}.
$$
 In \cite{GLM23}, we have considered a uniform magnetic field which is perpendicular to the 2D material axis. The reduction here is obtained considering operators commuting with \emph{magnetic translations}. We first provide a decomposition of such operators and then state the equivalent one-dimensional problem in which the kinetic energy per unit surface $\frac12\underline{\Tr}_{2D}(\mathrm{L}_\mathrm{A} \gamma)$  has been repalaced by $$ \frac12\Tr(-\Delta_1 G)+ \Tr(F(b,G)),
 $$  
 where $\mathrm{L}_\mathrm{A}$ denotes the three-dimensional Landau operator, $b$ the magnetic field strenghten, and $F(b,\cdot)$ a suitable piecewise linear function,  see for instance \cite[Theorem 3.1]{GLM23}.
 
\medskip
In the present paper, we consider homogeneous 1D materials, i.e. nuclear charge distributions of the form 
$$
\mu(x_1,x_2,x_3)=\mu(x_1,x_2)
$$
that do not depend on the third variable $x_3$, the axis of the material. As we consider convex models, we assume that the electronic density matrices follow the same symmetries
\begin{equation}\label{eq:1D-symmetry}
	\tau_R \gamma=\gamma \tau_R,\quad \forall R\in\RR,
\end{equation}
where $\tau_Rf(x_1,x_2,x_3)=f(x_1,x_2,x_3-R)$ denotes the translation operator by the vector $Re_3$. Again, states satisfying the symmetry~\eqref{eq:1D-symmetry} have densities that depend only  on $x_1$ and $x_2$ and the trace per unit length  of $\gamma$ can be defined by  $\underline{\Tr} (\gamma)=\int_{\RR^2}\rho_\gamma$. In the reduced Hartree-Fock model, the energy of the system is 
\begin{equation}\label{eq:3D}
\underline{\cE}(\gamma)=\frac12\underline{\Tr} \bra{-\Delta_3 \gamma}+ \tilde{\cD}_2(\rho_\gamma-\mu), 
\end{equation}
where $\tilde{\cD}_2$ is the 2D Hartree interaction. Formally, the interaction term $\cD_2$ is defined by
\begin{equation}\label{eq:2D-hartree}
\tilde{\cD}_2(f,g)=-2\int_{\RR^2\times \RR^2}{f(x)}\log(\av{x-y})g(y). 
\end{equation}
This definition is only valid for functions that decay fast enough at infinity. Besides, it is unclear whether this bilinear form $\tilde{\cD}_2$ is bounded from below or not.
One key result of this paper is a regularization of this expression adapted to neutral systems; and a definition of the corresponding mean-field potential (see Section~\ref{sec:Hartree}). 

For the kinetic part, following the same arguments as in~\cite{GLM21}, 
we prove that this problem is equivalent to a two-dimensional problem,
where the three-dimensional kinetic energy $\frac12\underline{\Tr} (-\Delta_3\gamma)$
is replaced by a two-dimensional kinetic energy
of the form
$$
\frac12\Tr(-\Delta_2 G)+\frac{\pi}{6} \Tr(G^3), 
$$
defined on $\cG:=\{ G\in \bS(L^2(\R^2)) : G\ge0,\; \Tr(G)<\ii\} $. The reduced 2D rHF energy 
then reads
\begin{equation}\label{eq:1D}
\cE(G)=\frac12\Tr(-\Delta_1 G)+ \frac{\pi}{6}\Tr(G^3)+ \cD_2(\rho_G-\mu). 
\end{equation}
We show that this problem of minimizing the energy \eqref{eq:1D} over a suitable set of admissible states admits a unique minimizer, called ground state. We also derive the self-consistent equation that this ground state needs to satisfy.  

Finally, we perform numerical simulations on the simpler Thomas--Fermi model, where the kinetic energy is expressed as an explicit functional of the density $\int \rho^{5/3}$. This allows us to exhibit the properties of the regularized 2D Hartree interaction with respect to the classical expression~\eqref{eq:2D-hartree}.

\medskip
This article is organized as follow. Section~\ref{sec:Hartree} is devoted to the regularization of the 2D Hartree interaction and the definition of the corresponding mean-field potential. In Section~\ref{sec:kin-ene}, we show how to reduce the 3D problem~\eqref{eq:3D} into the 2D model~\eqref{eq:1D}, that we study in Section~\ref{sec:reduced}. Finally, in Section~\ref{sec:num},  we consider the orbital free Thomas-Fermi model as  a semi--classical limit of the rHF model in order to perform elementary numerical illustrations.\\
\linebreak 
\textbf{Acknowledgments.} The research leading to these results has received funding from OCP grant AS70 ``Towards phosphorene based materials and devices''. S. Lahbabi thanks the CEREMADE for hosting her during the final writing of this article.

\section{2D Hartree interaction}\label{sec:Hartree}
For 1D crystals, the Hartree interaction kernel has been derived in~\cite{BlancLeBr00} by means of thermodynamic limit procedure. It is given by 
$$
G(x)=-2\pi \log(\av{(x_1,x_2)})+\sum_{k\in \ZZ}\bra{\frac{1}{\av{x-ke_3}}-\int_{-1/2}^{1/2}\frac{dt}{\av{x-(y,k)e_3}}},\quad x\in\R^3,
$$
and the Hartree interaction energy is formally given by
$$
\cD_\per(f,g)=\int_{\Gamma\times \Gamma} f(x)G(x-y)g(y),
$$
where $\Gamma=\RR^2\times \com{-1/2,1/2}$. 
For functions $f,g$ which do not depend on the third variable $x_3$, this interaction reduces to 
\begin{equation}\label{eq:hartree-non-reg}
\cD_2(f,g)=-2\int_{\RR^2\times \RR^2}{f(x)}\log(\av{x-y})g(y).
\end{equation}
The integral in~\eqref{eq:hartree-non-reg} is only valid for functions decaying fast enough and the map $f\mapsto \cD_2(f)$ is not convex in general. 
The aim of this section is to give a regularization of this expression suitable for neutral charge distributions, which, in our context, will be $f=\rho -\mu$.  We therefore adopt the following definition 
\begin{equation}
	\cD(f,g) := 4\pi \int_{\R^2} \dfrac{\overline{\cF_2(f)(k)}\cF_2(g)(k)}{|k|^2} \rd k,
\end{equation} 
where $f$ and $g$ belong to the following Coulomb space
\[
\cC := \left\{  f\in L^1(\R^2),\; k\mapsto\frac{\cF_2(f)(k)}{|k|} \in L^2(\R^2)  \right\}. 
\]
We have denoted by $\cF_2$ the Fourier transform in $\RR^2$ defined by $\cF_2f(k)=\frac{1}{2\pi}\int_{\RR^2}f(x)e^{-ikx}dx$. We note  that the integrability of $k\mapsto\frac{|\cF_2(f)(k)|^2}{|k|^2}$ over $\R^2$ implies that 
$$\cF_2(f)(0)=\int_{\R^2} f(x) \rd x =0,$$
which means that elements of $\cC$ are neutral distributions. Conversely, for neutral distributions that decay fast enough, $\cD_2$ and $\cD$ coincide.

\begin{lemma}
	If $f\in L^1(\RR^2)$ is such that $\int f=0$ and $\log(1+\av{\cdot})f\in L^1(\RR^2)$, then 
	$$
	\cD(f)=\cD_2(f). 
	$$
\end{lemma}
\noindent
The proof of this lemma can be found in \cite{FreMulTho22}. 

\medskip\noindent
For $f\in \cC$, let us introduce $\rW_f $ as
\[  \rW_f := \cF_2^{-1}\left(k\mapsto\frac{\cF_2(f)(k)}{|k|}\right)= (-\Delta_2)^{-\frac{1}{2}} f\quad \in L^2(\R^2). 
\]
Thanks to Parseval identity, one can rewrite the interaction $\cD$ as 
\[  \cD(f,g) =  4\pi \int_{\R^2} \overline{\rW_f (x)}{W_g(x)} \rd x,\qquad\forall f,g\in \cC.
\]

\medskip\noindent
An explicit expression of $\rW_f$, for $f\in\cC$, is given through the Riesz potential as follows (see for instance \cite[Chapter V]{Stein19}). 
\begin{lemma}
	Let $f\in \cC$. Then,
	\begin{equation}\label{eq Wf=Riesz pot}
		\rW_f (x) = \frac{1}{2\pi}\int_{\R^2}\frac{f(y)}{|x-y|} \rd y.
	\end{equation}
\end{lemma}

The integral in \eqref{eq Wf=Riesz pot} is absolutely convergent whenever $f\in L^p(\R^2)$, for some $1\le p<2$ (see for instance \cite[Theorem 1 of Chapter V]{Stein19}). Furthermore, if $p\in(1,2)$ and $q=\frac{2p}{2-p}$. Then,
	\begin{equation}\label{eq |W_f|_q < |f|_p }
		\| \rW_f \|_q \le C \|f\|_p,\qquad \forall f \in L^p(\R^2),
	\end{equation} 
	for a suitable constant $C$ which is independent of $f$.


\medskip\noindent
We turn now to the definition of a \emph{mean-field potential} $\rV_f$  associated to a charge distribution $f\in\cC$. It needs to satisfy 
$$
\cD(f,g) = \int_{\R^2} \rV_f (x) g(x)\rd x,\qquad \forall g\in\cC
$$
and 
$$
 -\Delta \rV_f = 4\pi f 
$$
in a weak sense. One would consider $\rV_f=4\pi (-\Delta)^{-1}f = 4\pi W_{W_f}$. However,  the integral in \eqref{eq Wf=Riesz pot} is not necessarily convergent in the critical case of $p=2$. We suggest an alternative way to define the potential $\rV_f$, for $f\in\cC$ through a modification of the {Riesz potential} suggested in \cite{Kuro88}. Let us define the following two functions
\[ q_{1}(x) = |x|\left(1+\log^+\left(\frac{1}{|x|}\right)\right)\quad \text{and}\quad  q_{2}(x) = |x|\left(1+\log^+\left(|x|\right)\right),
\]
where $$\log^+(|x|):=\log(|x|)\1_{\{|x|\ge1\}}=\begin{cases} \log(|x|)\quad\text{if}\quad |x|\ge 1 \\ 0 \quad\text{else} \end{cases}.$$
In the following proposition, we define and state some properties of the modified Riesz potential. For the proof and further details, we refer to \cite[Corollary 3.23]{Kuro88}.
\begin{proposition}\label{prop U^f well defined}
	For $g\in L^2(\R^2)$, define
	\begin{equation} \label{eq modified Riesz transf}
		\mathcal{I}_{1}^{g}(x) := \frac{1}{2\pi}\int_{|y|<1}^{} \frac{g(y)}{|x-y|} \; \rd y \quad \text{and} \quad \mathcal{I}_{2}^{g}(x) := \frac{1}{2\pi}\int_{|y| \ge 1}^{} \left( \frac{1}{|x-y|} - \frac{1}{|y|} \right) g(y) \; \rd y.
	\end{equation}
	One has the following
	\begin{enumerate}[a)]
		\item $\cI_1^g$ and $\cI_2^g$ are well-defined, i.e., the integrals in \eqref{eq modified Riesz transf} are convergent.
		\item There exist two real numbers $C_{1}$ and $C_{2}$ such that
		\begin{align}
			\left\|  \frac{\mathcal{I}_{1}^{g} }{q_{1}} \right\|_{L^{2}} \;\; \le\; C_{1} \left\|  g\right\| _{L^{2}} \quad and \quad \left\| \frac{\mathcal{I}_{2}^{g} }{q_{2}}\right\| _{L^{2}} \;\; \le\; C_{2} \left\| g\right\|	_{L^{2}}.
			\label{eq L2 Estimate for U/q}
		\end{align}
	\end{enumerate}
\end{proposition}
\noindent
As a consequence of Proposition~\ref{prop U^f well defined}, we have $\mathcal{I}_i^f\in L^1_\loc(\RR^2)$, for $i\in \set{1,2}$. Now, for $f\in\cC$, we define the mean field potential as
\begin{equation}\label{eq V=U1+U2}
	\rV_f : =  4\pi \left(\mathcal{I}_{1}^{\rW_f} + \mathcal{I}_{2}^{\rW_f} \right).
\end{equation}
The following theorem is one of the main results of the paper and it allows to rewrite the potential energy $\cD$ in terms of the mean--field potential $V_f$. 

\begin{theorem}\label{thm. pot energy = int pot. density} 
Let $f\in \cC$. Then 
\begin{equation}\label{energy_by_potential}
	\mathcal{D}(f,h) = \int_{\mathbb{R}^{2}} \rV_{f}(x)\, h(x) \rd x\qquad \forall h \in \mathcal{C},				
\end{equation}
			and  $\rV_{f}$ satisfies 
			\begin{equation}\label{eq:laplace}
			-\Delta\rV_{f} =4 \pi f  
		\end{equation} 
			in a distributional sense.
\end{theorem}

\begin{proof}
Let us first prove~\eqref{energy_by_potential} for  $h$ in the Schwartz space of fast decaying functions $\mathcal{S}(\RR^2)$. 
Since
$\cI_j^{\rW_{f}} h = (q_j^{-1}\cI_j^{\rW_{f}})\, q_j h  \in L^1(\R^2)$, for $j\in\{1,2\}$. Then, one can write
\begin{align}\label{eq loc proof}
	\nonumber   \frac{1}{4 \pi}\int_{\R^2}\rV_{f}(x)\,h(x)\,\rd x :=& \int_{\R^2}\mathcal{I}_1^{\rW_{f}}(x) h(x)\, \rd x + \int_{\R^2} \mathcal{I}_{2}^{\rW_{f}}(x) h(x)\, \rd x \\
	\nonumber    =& \frac{1}{2\pi}\int_{\R^2} \int_{\{|y|<1\}} \frac{\rW_{f}(y)}{|x-y|} h(x)\, \rd y \rd x \\
	&+ \frac{1}{2\pi} \int_{\R^2} \int_{\{|y|\ge 1\}} \left(\frac{1}{|x-y|}-\frac{1}{|y|}\right)\rW_{f}(y) h(x)\, \rd y \rd x.
\end{align}
In order to apply Fubini theorem and change the order of the integrals in the above double intgrals, we have to check the absolute integrability. 
Regarding the first term, one can note that $(|q_{1}|)^{-1}\mathcal{I}_{1}^{|\rW_{f}|}$ is in $L^2(\R^2)$, thanks to \eqref{eq L2 Estimate for U/q}, and so is $q_{1}h$. Therefore, one has
\begin{align}
	\nonumber  \frac{1}{2\pi}\int_{\R^2} \int_{\{|y|<1\}} \frac{\rW_{f}(y)}{|x-y|} h(x) \,\rd y \rd x &= \frac{1}{2\pi}\int_{|y|<1} \int_{\mathbb{R}^{2}} \frac{\rW_{f}(y)}{|x-y|} h(x) \,\rd x\rd y  \\
	&= \int_{\{|y|<1\}} \rW_{f}(y) \rW_{h}(y) \;\rd y.
	\label{S1}
\end{align}
Taking into the account that the estimate in \eqref{eq L2 Estimate for U/q} holds true if one substitutes the kernel defining $\mathcal{I}_{2}$ by its absolute value (see for instance \cite[Theorem 3.19]{Kuro88}), then the last double intergal of~\eqref{eq loc proof} is absolutely convergent as well. Thus,
\begin{align}
\nonumber  & \frac{1}{2\pi}\int_{\R^2} \int_{|y|\ge 1} \left(\frac{1}{|x-y|}-\frac{1}{|y|}\right) \rW_{f}(y) h(x) \,\rd y \rd x   \\
\nonumber& \qquad\qquad = \frac{1}{2\pi}\int_{|y|\ge 1} \int_{\R^2} \left(\frac{1}{|x-y|}-\frac{1}{|y|}\right)\rW_{f}(y) h(x) \,\rd x \rd y \\
\nonumber& \qquad\qquad = \int_{|y|\ge 1} \rW_{f}(y) \left( \rW_{h}(y) - \frac{1}{2\pi |y|} \int_{\mathbb{R}^{2}} h \right)\, \rd y  \\
&\qquad\qquad= \int_{|y|\ge 1}^{} \rW_{f}(y) \rW_{h}(y) \,\rd y.
\label{S2}
\end{align}
By adding \eqref{S1} et \eqref{S2}, we obtain
\[
\int_{\mathbb{R}^{2}}^{}\rV_{f}(x)\,h(x)\,dx \: = 4\pi \int_{\mathbb{R}^{2}}^{} \rW_{f}(x) \rW_{h}(x) \,dx.
\]
Now, let $h\in\cC$. Then, $h\in L^1(\R^2)$ and $\rW_h \in L^2(\R^2)$. By density of the Schwartz space $\cS$, one can find $(h_n)_n\subset \cS$ with $\int_{\R^2} h_n \, \rd x=0$ such that $h_n\to h$ in $L^1(\R^2)$ and $\frac{\cF_2(h_n)}{|\cdot|}\to \frac{\cF_2(h)}{|\cdot|}$ in $L^2(\R^2)$. Hence, one obtains \eqref{energy_by_potential}, for $h\in\cC$, by letting $n\to\ii$.

\medskip\noindent
We move now to the proof of~\eqref{eq:laplace}. If one denotes by $(\cdot,\cdot)$ the distributional brakets, then,
for $\phi\in C^\ii_c(\RR^2)$, 
\begin{align*}
	( -\Delta_2 V_f,\phi ) =(  V_f,-\Delta_2\phi)=\int_{\R^2} V_f \Delta_2\phi \,\rd x.
\end{align*}
One knows that $h:=-\Delta \phi$ satisfies $\int h=0$. Thus, by the previous result
\begin{align*}
		( -\Delta_2 V_f,\phi )= \cD(f,h)&= 4\pi \int_{\RR^2}\frac{\overline{\cF_2(f)(k)}\cF_2(h)(k)}{\av{k}^2}\rd k\\
		&=4\pi \int_{\RR^2}\overline{\cF_2(f)(k)}\cF_2(\phi)(k)\rd k=\langle 4\pi f,\phi \rangle.
\end{align*}

\end{proof}
\section{Model reduction}\label{sec:kin-ene}
In this section, we show that the 3D reduced Hartree Fock model for homogeneous 1D crystal is equivalent to a 2D model. Other DFT models can be handled similarly.  We recall that the nuclear charge distribution $\mu\in L^1_\loc(\R^3)$ satisfies
\begin{equation}
\mu\geq 0,\quad \text{and} \quad 	\mu(x_1,x_2,x_3)=\mu(x_1,x_2,x_3-R),\quad \forall R \in\R.
\end{equation} 
We assume that the charge per unit length is finite, that is,
\[ Z:= \int_{\R^2} \mu(x_1,x_2,0)\; \rd x_1 \,\rd x_2 < \ii.
\]
In Hartree--Fock type models, the state of electrons is described by a bounded self--adjoint operator $\gamma$ satisfying  Pauli exclusion principle  $0\le\gamma\le 1$. In our particular case, the system has the following features. 
\begin{itemize}
	\item It is equidistributed in one direction of the three--dimensional space,
	\item it is confined in the remaining two directions with a finite number of particles.
\end{itemize}
The electronic states $\gamma$  share the aforementioned invariance as follows
\begin{equation}\label{eq. intro gamma commute tau-R}
	\gamma\tau_{R}=\tau_{R}\gamma,\qquad\forall R\in\R,
\end{equation}
where $\tau_R$ denotes the translation operator by the vector $Re_3$ ($\tau_R f(x)=f(x_1,x_2,x_3-R)$). 
The 3D rHF energy of a state $\gamma$ in 
\[ \cK := \{  \gamma\in \underline{\fS}(L^2(\R^3)),\;  0\le\gamma\le 1,\;  \tau_{R}\gamma=\gamma\tau_{R},\;\forall R\in\R\}.
\]
is 
$$
\underline{\cE}(\gamma)=\frac12\underline{\Tr} \bra{-\Delta_3 \gamma}+ \cD(\rho_\gamma-\mu),
$$
where we have denoted by  $\underline{\fS}(L^2(\RR^3))$ the space of self--adjoint, locally trace class operators acting on $L^2(\RR^3)$, and $\cD$ is the regularized 2D Hartree interaction introduced in the previous section. We now define the 2D energy functional for a state in 
\begin{equation}\label{eq set of admissible states G}
	\cG := \{   G\in\fS(L^2(\R^2)),\; G\ge 0 
	 \},
\end{equation} 
where we have denoted by  ${\fS}(L^2(\RR^3))$ the space of self--adjoint, trace class operators acting on $L^2(\RR^2)$, by 
\begin{equation}\label{eq:energy-2D}
\cE(G)=\frac12\Tr(-\Delta_1 G)+ \frac{\pi}{6}\Tr(G^3)+ \cD(\rho_G-\mu). 
\end{equation}
The main result of the section is the following theorem. 

\begin{theorem}\label{th:equivalence}
The ground state energies 
$$
\underline{I}=\inf\set{\underline{\cE}(\gamma),\; \gamma\in \cK}
$$	
and
\begin{equation}\label{eq:2D-prob}
I=\inf\set{\cE(G),\; G\in\cG}
\end{equation}
are equal and the minimizers of both energies share the same densities, which does not depend on $x_3$. 
\end{theorem} 

The proof of theorem~\eqref{th:equivalence} follows the same line as the one of ~\cite[Theorem 2.8]{GLM21}. We highlight here some arguments in the 2D case. 

\noindent
\emph{Diagonalization of admissible states.} We start by showing that states in $\cK$ can be diagonalized. Let us introduce the partial Fourier transform
\begin{align*}
	\cU :  L^2(\R^3) & \to L^2(\R,L^2(\R^2)) \\
	f& \mapsto \cF_1 ( f(x_1,x_2,\cdot))(k),
\end{align*}
where $\cF_1$ refers to the Fourier transform in $L^1(\RR)$ $\cF_1 ( f(x_1,x_2,\cdot))(k) = \frac{1}{\sqrt{2\pi}} \int_\R \re^{-\ri k t}f(x_1,x_2,t)\rd t$. According to the Floquet--Bloch decomposition (see \cite[Section XIII–16]{ReedSimon4}), for any $\gamma\in\cK$ there exists a unique family of self--adjoint operators $(\gamma_k)_k$ on $L^2(\R^2)$, called the fibers of $\gamma$, such that
\begin{equation}\label{eq gamma and gamma-k}
	\cU \gamma \cU^{-1} = \int_\R^\oplus\gamma_k \rd k,
\end{equation}
that is, for any $f=(f_k)_{k\in\R}\in L^2(\R,L^2(\R^2))$,
\[  (\cU \gamma \cU^{-1} f)_k =  \gamma_k f_k .
\]
 We point out that, if $\gamma\in\cK$, then its fibers $\gamma_k$  satisfy Pauli principle $0\le\gamma_k\le1$. The following proposition links other  properties of $\gamma\in \cK$ to those of its fibers.
\begin{proposition}\label{prop elemntary prop of trans inv op}
	Let $\gamma= \int_\R^\oplus\gamma_k\in\cK$ and suppose that $\gamma$ is locally trace class. One has 
	\begin{itemize}
		\item  $\gamma$ admits an integral kernel $\gamma(\cdot;\cdot):(\R^3)^2 \to \R$ and
		\begin{equation}\label{eq kernel gamma vs kernel gamma-k}
			\gamma(x_1,x_2,x_3;y_1,y_2,y_3) = \frac{1}{{2\pi}} \int_\R \re^{\ri k(x_3-y_3)} \gamma_k (x_1,x_2;y_1,y_2) \rd k,
		\end{equation}
		where, for every $k\in\R$, $\gamma_k(\cdot;\cdot):\R^2\to\R$ is the integral kernel of $\gamma_k$.
		\item Let $\rho_\gamma$ denote the density of $\gamma$, given by 
		$ \rho_\gamma(x)=\gamma(x,x)$ for all $x\in\R^3$. Then,
		\[ \rho(x_1,x_2,x_3)= \frac{1}{{2\pi}} \int_\R \rho_k (x_1,x_2;x_1,x_2) \rd k,
		\]
		where $\rho_k$ refers to the density of $\gamma_k$.\\
		In particular, $\rho_\gamma$ does not depend on the third variable $x_3$.
	\item The \emph{trace per unit length } of $\gamma$ is given by 
	$$
				\VTr(\gamma) :=\int_{\RR^2} \rho_\gamma= \frac{1}{{2\pi}}\int_\R \Tr(\gamma_k) \rd k  .
		$$
		
	\end{itemize}
\end{proposition}
\begin{proof}
Since $\gamma$ is locally trace class, then $\gamma_k$ is trace class, for all $k\in\R$. 
By straightforward computation, one gets for a smooth enough $f$ and $(x_1,x_2,k)\in\R^3$
\begin{align*}
	(\cU\gamma f)_k (x_1,x_2) &= \frac{1}{\sqrt{2\pi}}\int_\R \re^{-\ri kt}\,(\gamma f)(x_1,x_2,t) \rd t\\
	&=\frac{1}{\sqrt{2\pi}}\int_\R \int_{\R^3} \re^{-\ri kt} \gamma(x_1,x_2,t ; y_1,y_2,\tilde{t}) f(y_1,y_2,\tilde{t})\; \rd y_1 \rd y_2 \rd\tilde{t} \rd t,
\end{align*}
On the other hand,
\begin{align*}
	(\gamma_k(\cU f)_k) (x_1,x_2) &= \int_{\R^2} \gamma_k (x_1,x_2; y_1,y_2) (\cU f)_k (y_1,y_2) \rd y_1 \rd y_2 \\
	&= \frac{1}{\sqrt{2\pi}}\int_{\R^2}\int_\R \re^{-\ri k\tilde{t}} \gamma_k (x_1,x_2; y_1,y_2) f(y_1, y_2, \tilde{t}) \rd y_1 \rd y_2 \rd \tilde{t}
\end{align*}
By identification,
\[ \gamma_k (x_1,x_2; y_1,y_2) = \int_\R \re^{-\ri k(t-\tilde{t})} \gamma(x_1,x_2,t ; y_1,y_2,\tilde{t}) \rd t. 
\]
Therefore, applying an inverse Fourier transform, one obtains \eqref{eq kernel gamma vs kernel gamma-k}. The claim on the density follows immediately by \eqref{eq kernel gamma vs kernel gamma-k}.
\end{proof}

\noindent
\emph{Diagonalization of the kinetic energy}
Let $\gamma= \int_\R^\oplus\gamma_k\in\cK$. 
Let $f=(f_k)_{k\in\R}\in L^2(\R,L^2(\R^2))$ such that $f_k$ is smooth enough for any $k$. Straightforward computation yields
\[ (\cU(-\Delta_3) \gamma\cU^{-1}f)_k = (-\Delta_2 + k^2) \gamma_k f_k.
\] 
Hence, the kinetic energy per unit length of $\gamma$ can be written 
\[ \frac12 \VTr(-\Delta_3 \gamma) =\frac{1}{{4\pi}}\int_\R \Tr((-\Delta_2 + k^2) \gamma_k)\, \rd k.  \]
\noindent
\emph{Reduced states}
Now, we map any admissible state $\gamma$ to a reduced 2D state $G$ as follows
\begin{align*}
	\Lambda : \cK & \to \cG \\
	\gamma & \mapsto \Lambda(\gamma):=\frac{1}{2\pi}\int_\R \gamma_k \rd k. 
\end{align*}
We notice that elements of $\cG$, unlike those of $\cK$, are compact operators. 
We have the following result. 
\begin{proposition}\label{prop Lambda is onto}
	The map $\Lambda$ is onto. Furthermore, for any $\gamma\in\cK$, 
	\[ \rho_\gamma(x_1,x_2,x_3) = \rho_{\Lambda(\gamma)}(x_1,x_2),\qquad \forall(x_1,x_2)\in\R^2.
	\]
\end{proposition}

\begin{proof}
Let $G\in\cG$. One can write $G=\sum_j g_j |\varphi_j\rangle\langle \varphi_j|$, with $(\varphi_j)_j$ an orthonormal family in $L^2(\R^2)$ and $g_j\geq 0$ such that $\sum_j g_j=\Tr(G)<\ii$. We construct $\gamma$ satisfying $\Lambda(\gamma)=G$ through its fibers $(\gamma_k)_k$. Let us define, for any $k\in\R$, $\gamma_k$ as 
	$$\gamma_k=\sum_j \alpha_j^{(k)}|\varphi_j\rangle\langle \varphi_j|\quad \text{with}\quad \alpha_j^{(k)}=1_{\{|k|\le s_j\}},$$
	where, for any $j\in\N$, $s_j\geq 0$ will be determined later. One has $0\le\gamma_k\le 1$ and $\gamma_k\in\fS(L^2(\R^2))$. Thus, $\gamma\in\cK$. Moreover, $\rho_{\gamma_k} = \sum_j \alpha_j^{(k)} |\varphi_j|^2$. Therefore,
	\begin{align*}
		\rho_\gamma(x) &= \frac{1}{2\pi}\int_\R \rho_{\gamma_k} \rd k\\
		& = \frac{1}{2\pi}\sum_j \int_\R 1_{\{|k|\le s_j\}} |\varphi_j(x_1,x_2)|^2 \rd k = \frac{1}{\pi} \sum_j s_j |\varphi_j(x_1,x_2)|^2 .
	\end{align*}
	Finally, with the choice $s_j=\pi{g_j}$, one obtains $\rho_\gamma=\rho_G$, which proves the claim.
\end{proof}

\noindent
\emph{Reduction of the kinetic energy}
We now rewrite the kinetic energy as a  functional of elements of $\cG$. 
\begin{theorem}\label{thm E_kin gamma > Tr Delta G +...}
	Let $G\in\cG$. Then,
	\begin{equation}
		\inf_{\Lambda(\gamma)=G} \underline{\cE}^\kin (\gamma) =\frac12 \Tr(-\Delta_2 G) +\frac{\pi^2}{6} \Tr(G^3).
	\end{equation}
\end{theorem}  

\begin{proof}
	Let $G\in\cG$ and $\gamma=\int^{\oplus}_\RR\gamma_k\in\cK$ with $\Lambda(\gamma)=G$. Write $G=\sum_j g_j |\varphi_j\rangle\langle \varphi_j|$. We have
	\begin{align*}
		\underline{\cE}^\kin (\gamma) &= \frac{1}{2}\VTr(-\Delta_3 \gamma)=\frac{1}{4\pi}\int_\R \Tr((-\Delta_2 + k^2) \gamma_k)\, \rd k\\
		&=\frac12 \Tr(-\Delta_2 G) + \frac{1}{4\pi} \int_\R \sum_j k^2 m_{j,k} \rd k,
	\end{align*}
	where $m_{j,k} :=\langle \gamma_k \varphi_j ,\varphi_j\rangle$. Notice that
	\begin{equation}\label{eq bathtube cond}
		0\le m_{j,k} \le 1 \quad\text{and}\quad
		\displaystyle\int_\R m_{j,k}\, \rd k= 2\pi g_j.
	\end{equation}
	Hence,
	\begin{align*}
		\underline{\cE}^\kin (\gamma)\ge \frac12 \Tr(-\Delta_2 G) +\frac{1}{4\pi} \sum_j \inf\left\{ \int_\R k^2 m(k) \rd k,\;  m(k)\;  \text{satisifying}\; \eqref{eq bathtube cond} \right\}.
	\end{align*}
	According to the Bath-tub principle, \cite[Theorem 1.14]{LiebLossAnalysis}, the infinimum in the RHS of the above inequality is achieved by $m_j^*(k):=1_{\{|k|<\pi g_j\}}$, the same state exhibited in the proof of Proposition~\ref{prop Lambda is onto}. 
	Therefore,
	\begin{align*}
		\underline{\cE}^\kin (\gamma)&\ge 
		 \frac12 \Tr(-\Delta_2 G) +\frac{1}{6\pi} \sum_j (\pi g_j)^3=  \frac12 \Tr(-\Delta_2 G) +\frac{\pi^2}{6}\Tr(G^3),
	\end{align*}
	with equality for $\gamma^*\in\cK$ with $\gamma_k^*=\sum_j m_j^*(k) |\varphi_j\rangle\langle \varphi_j|$. 
\end{proof}

\begin{remark}
	A consequence of Theorem \ref{thm E_kin gamma > Tr Delta G +...} is that for a representable charge density $\rho\in L^1(\R^2)$, one could associate the kinetic energy functional as follows
	\begin{equation}\label{eq kin energy as function of rho}
		\cE^\kin(\rho)= \inf\left\{ \Tr(-\Delta_2 G) +\frac{2\pi^2}{3} \Tr(G^3),\; G\in\cG_0,\; \rho_G=\rho \right\}.
	\end{equation}
	Finding an explicit expression for the above functional $\rho\mapsto\cE^\kin(\rho)$ is highly challenging. It is only possible to associate a lower bound to this energy as follows
	\begin{equation}\label{eq Lieb-Thirr}
		C \left( \int_{\R^2} \rho_G^{5/3}\right) \le \frac12\Tr(-\Delta_2 G)+\frac{\pi^2}{6}\Tr(G^3)=\cE^\kin(\rho).
	\end{equation}
	for an appropriate constant $C$. This is a consequence of the Lieb--Thirring inequality, see for instance \cite{LieThi-75,LieThi-76} and \cite[Appendix A]{GLM21} (see Section~\ref{sec:TF}). In particular, this yields that $\rho_G \in L^{5/3}$, once the kinetic enegry of $G$ is finite. Actually, the two dimensional Lieb--Thiring inequality applied to $G/Z$ implies that $\rho_G\in L^2(\R^2)$ and the Hoffman-Ostenhof inequality claims that $\sqrt{\rho_G}\in H^1(\R^2)$, see \cite{Hoffman_Ostenhof77}.
\end{remark}

The proof of Theorem~\ref{th:equivalence} can be easily deduced from Theorem~\ref{thm E_kin gamma > Tr Delta G +...} and  Proposition~\ref{prop Lambda is onto}.

\section{Well-posdness of the reduced model}\label{sec:reduced}

This section is devoted to the analysis of the reduced  two-dimensional problem~\eqref{eq:2D-prob}. In particular, we prove well-posdness and show that the ground state satisfies a self-consistent equation. 

\begin{theorem}\label{th:existence}
 Assume that $\mu\in L^2(\RR^2)\cap L^1(\RR^2)$.  Then the energy $\cE $ in \eqref{eq:energy-2D} admits a unique minimizer $G_*$ in $\cG$. Furthermore, one has
\begin{equation}\label{eq Eul-Lar equat for G_*}
\begin{cases}
\quad G_* =  \frac{\sqrt{2}}{\,\pi}\left(\lambda_* - H_*\right)_+^{\frac{1}{2}} \\
\quad H_* = -\Delta_2 + \rV_{\rho_*-\mu}
\end{cases},
\end{equation}
where $\rho_*:=\rho_{G_*}$, $x_+:=\max(0,x)$ for $x\in\R$ and $\left(\lambda_* - H_*\right)_+^{\frac{1}{2}}$ should be unterstood in terms of functional calculus for the self--adjoint operator $H_*$.
\end{theorem} 

\begin{proof}
The proof of this theorem follows the same lines as the one of~\cite[Theorem 2.8]{GLM21}. We detail here the main arguments in our 2D case. 
\noindent
\emph{Existence and uniqueness of the minimizer}. The set $\cC$ is convex and the functional $\cE $ is convex and positive. Then, the existence of a minimizer is reduced to the lower semi-continuity of $\cE $. Let
	let $(G_n)_n\subset\cG$  be a minimizing sequence of $\cE $ and denote by $\rho_n:=\rho_{G_n}$. One has
	\[ \cE (G_n) = \frac{1}{2}\Tr(-\Delta_2 G_n) + \frac{\pi^2}{6} \Tr(G_n^3) + \cD(\rho_{n} -\mu) \le C,
	\]
	where $C$ is a positive constant. In particular, $(G_n)_n$ and $(-\Delta G_n)_n$ are bounded sequences in the Schatten space of trace class operators $\fS(L^2(\R^2))$ and $(G_n)_n$ is bounded in the Schatten space of cubic trace class operators $\fS_3(L^2(\R^2))$. Hence, taking into the account that $-\Delta_2$ is closed, there exists $G_*\in\fS(L^2(\R^2))\cap \fS_3(L^2(\R^2))$ such that
	$$(G_n,-\Delta_2 G_n)\rightharpoonup (G_*,-\Delta_2 G_*) \quad\text{in}\quad \left(\fS(L^2(\R^2))\right)^2,$$
	$$G_n\rightharpoonup G_*\quad\text{in}\quad \fS_3(L^2(\R^2)),$$
	and 
	\[ \frac{1}{2}\Tr(-\Delta_2 G_*) + \frac{2\pi^2}{3} \Tr(G_*^3) \le \liminf_n \left( \frac{1}{2}\Tr(-\Delta_2 G_n) + \frac{2\pi^2}{3} \Tr(G_n^3)\right).
	\]
	By the Lieb-Thirring inequality~\cite{LieThi-75}, $(\rho_n)_n$ is bounded in $L^{2}(\R^2)$. Thus, $(\rho_n)_n$ converges weakly in $L^2(\R^2)$. Moreover, its weak limit is $\rho_*:=\rho_{G_*}$. Indeed, for any function $h\in C_c(\R^2)$ one has
	\begin{align*}
		\int_{\R^2} h\rho_*  \rd x &= \Tr\left((1-\Delta)^{-1}\cM_h G_*(1-\Delta)\right)\\
		& = \lim_n  \Tr\left((1-\Delta)^{-1}\cM_h G_n (1-\Delta)\right) = \lim_n \int_{\R^2} h\rho_n , 
	\end{align*}
	where $\cM_h$ stands for the multiplication operator by $h$ in $L^2(\R^2)$, and the convergence is obtained by the fact that $(1-\Delta)^{-1}\cM_h$ is a compact operator. 
	Besides, $\|\rW_{\rho_n -\mu}\|_2 =\cD(\rho_n-\mu)\le C$, for all $n\in\N$. Hence, $$W_{\rho_n-\mu}\rightharpoonup W_* \text{ in } L^2(\R^2)\quad \text{and}\quad \|W_* \|_2 \le\liminf_n \|\rW_{\rho_n -\mu}\|_2.$$
 Let us prove that $W_*=W_{\rho_*-\mu}$. Let $\varphi$ be a test function  in the Schwartz space $\cS(\R^2)$
 such that $\int_{\R^2} \varphi\, \rd x=0$. By Parseval's identity, one gets
	\begin{align*}
		\int_{\R^2} \rW_{\rho_n -\mu} (x) \varphi(x) \rd x = \left\langle \frac{\cF_2(\rho_n -\mu)}{|k|},\cF_2(\varphi) \right\rangle
		&= \left\langle \cF_2(\rho_n -\mu),\frac{\cF_2(\varphi)}{|k|} \right\rangle
		\\
		&=\langle \rho_n -\mu, \rW_\varphi \rangle.
	\end{align*}
	Letting $n$ go to $\infty$, and noting that $\rho_n-\mu\in L^2(\RR^2)$, we obtain 
	\[ \int_{\R^2} \rW_{*} (x) \varphi(x) \rd x = \int_{\R^2} \left(\rho_* -\mu \right)(x)\, \rW_\varphi (x)\, \rd x = \int_{\R^2} \rW_{\rho_* -\mu} (x) \varphi(x) \rd x.
	\]
	This actually identifies $\rW_*$ with $\rW_{\rho_* -\mu}$. Summarizing, one has $G_*\in\cG$ and 
	\[ I\le \cE (G_*) \le \liminf_n \cE (G_n) =I,
	\]
	which proves that $G_*$ is actually a minimizer of $\cE $. The uniqueness comes from the strict convexity of $\Tr\bra{G^3}$. 

\noindent
\emph{Proof of the Euler Lagrange equation.} Let $G_*$ denote the minimizer of $\cE $, $\rho_*:=\rho_{G_*}$ its density and $V_*:=V_{\rho_*-\mu}$ the mean-field potential. Let $G\in\cG$ and $t\in[0,1]$. One has $G_t:=G_*+t(G-G_*)\in\cG$ and $\cE (G_t)\ge \cE (G_*)$. Besides,
	\begin{align*}
		\cD(\rho_{G_t}-\mu) 
		&= \cD(\rho_*-\mu) +2t\cD(\rho_*-\mu,\rho_G -\rho_*) +o(t)\\
		&= \cD(\rho_*-\mu) + 2\, t\, \Tr(V_*(G-G_*)) +o(t).
	\end{align*}
Thus, 
	\[ t\left(\Tr\left(-\frac{1}{2}\Delta_2(G-G_*)\right) +\frac{\pi^2}{2}\Tr\left( G_*^2(G-G_*)\right) +  \Tr(V_*(G-G_*))\right) + o(t) \ge 0.
	\]
	Hence, dividing by $t$ and letting $t\to 0$, one obtains
	\[ \Tr\left(\left(-\frac{1}{2}\Delta_2 +\frac{\pi^2}{2} G_*^2+V_*\right)(G-G_*)\right) \ge 0.
	\]
	Proceeding as in \cite[Sect. 4.3.2]{GLM21} and \cite[Proposition 3.8]{GLM23}, set $H_*:=-\frac{1}{2}\Delta_2 +V_*$ and $L_*=H_*+\frac{\pi^2}{2} G_*^2$. Therefore, 
	$ \Tr(L_* G) \ge \Tr(L_* G_*)$, for all $G\in\cG$. This implies that $L_*$ is bounded from below and
	\[ \Tr(L_* G_*)=\inf\{ \Tr(L_* G),\; G\in\cG \}.
	\]
	Set $\lambda_*=\inf\sigma(L_*)$ the bottom of the spectrum of $L_*$. Then, $(L_* -\lambda_*)G_*=0$. Next, let us expand $G_*$ as follows $G_*=\sum_{j\in\N} g_j^* | \phi_j^*\rangle \langle \varphi_j^*| $, where $g_j\ge0$ for any $j\in\N$ with $\sum_j g_j= Z$, and $\overline{\{\varphi_j,\; j\in\N\}}=L^2(\R^2)$. One has $(L_*-\lambda_*)\varphi_j^* =0$, for any $j\in\N$ such that $g_j>0$. That is
	\[ H_*\varphi_j^* = \left(-\frac{1}{2}\Delta_2 +V_*\right)\varphi_j^* = \left(\lambda_* -\frac12(\pi g_j^*)^2\right)\varphi_j^*,\quad\text{once}\quad g_j>0.
	\]
	Hence, $s_j:=\lambda_* -\frac12(\pi g_j^*)^2$ belongs to the point spectrum of $H_*$. Conversely, if $s\in\sigma(H_*)$ such that $s<\lambda_*$, then $g=\frac{\sqrt{2}}{\pi}\sqrt{\lambda_*-s}$ is an eigenvalue of $G_*$. This yields $G_*=\frac{\sqrt{2}}{\pi}\left(\lambda_*-H_*\right)_{+}^{1/2}$.
Note that $H_*$ is a compact perturbation of $L_*$; thus they share the same essential spectrum. It follows that $H_*$ cannot have essential spectrum below $\lambda_*$. 
\end{proof}
\section{Numerical simulation}\label{sec:num}
\subsection{Thomas-Fermi model}\label{sec:TF}
We perform numerical illustrations on the simple Thomas-Fermi model, which can be seen as a semi--classical limit of the rHF model. 
The kinetic energy in this model reads
$$\cE^\TF_\kin (\rho)=c_\TF \int_{\R^2} \rho^{5/3}(x) \rd x,\quad \text{with}\quad c_{\TF}=\frac{3}{10}(3\pi^{2})^{2/3}.$$ 
Therefore, for a given $0\le\mu\in L^1(\R^2)$, we define the Thomas-Fermi energy per unit length as
\[ \cE^\TF(\rho) := c_\TF \int_{\R^2} \rho^{\frac{5}{3}}(x) \rd x + \frac{1}{2}\cD(\rho -\mu),
\]
for every $\rho\in \cC^\TF$, where
\[ \cC^\TF:=\{ \rho\in L^1(\R^2)\cap L^{5/3}(\R^2),\; \rho-\mu \in \cC \} .
\]
We point out that Thomas--Fermi model and its derivatives are widely considered in the literature, see \cite{Lieb77,CatLeBrLions98,BlancLeBr00}. The functional $\cE^\TF$ is stricltly convex and has a unique minimizer $\rho_\TF$ that satisifes the following self--consistent equation
\begin{equation}\label{eq EL for TF}
\begin{cases}
    \rho_{TF}=\left(\frac3{5c_{TF}}(\lambda-\rV_{\TF})_+\right)^{3/2}\\
    \rV_\TF = \rV_{\rho_\TF -\mu}
\end{cases},
\end{equation}
with $\lambda$ the Fermi level of the system is such that $\int_{\R^2}\rho_\TF=Z:=\int_{\R^2} \mu$.\\  We emphasize that we still use our expression for the potential energy $\cD$ and the mean--field potential in \eqref{eq V=U1+U2} which is involved in the self--consistent equation \eqref{eq EL for TF}. We next show a numerical solution to \eqref{eq EL for TF}. We also contrast the outcomes from this equation with the non-regularized reduced energy in which the Coulomb energy and potential are
\begin{equation}\label{eq tilde(D)=log(x-y)}
 {\cD}_2(f,g)=-2\int_{\R^2}\int_{\R^2} \log(|x-y|) f(x) g(y)\, \rd x \rd y
\end{equation}
and 
\begin{equation}\label{eq log pot}
\widetilde{\rV}_f (x)=-2 \int_{\R^2} \log(|x-y|) f(y) \, \rd y.
\end{equation}




\subsection{Numerical schemes}

Python is used to carry out the numerical simulation. The two-dimensional functions are evaluated on a two-dimensional grid representing $[-a,a]\times[-b,b]$, with $N_a\times N_b$ discretization points (we took $a=b=8$ and $N_a=N_b=41$). The various integrals inolved in the total energy are then approximated using quadrature techniques; we used the two-dimensional trapezoid rule. 

We use two different methods in the regularized and non-regularized case. In the latter,
the total energy is approximated using a numerical integration method. 
The constraints of the model are simply the positivity of the density $\rho$ along with the neutrality condition $\int_{\R^2} \rho=Z$. The resulting discretized problem is thus a finite-dimensional optimization problem on $\rho_{ij}$ (the approximated values of the density on the grid points) subject to linear constraints. This problem is solved numerically in Python using the predefined function $\texttt{minimize}$ from $\texttt{Scipy}$ library.
In the regularized case, we apply an iterative process to solve equation~\eqref{eq EL for TF} as follow:
at each iteration $n$, we evaluate the potential $\rV_n:=\rV_{\rho_n-\mu}$ using~\eqref{eq V=U1+U2} and we compute $\lambda_n$ solution to
\begin{equation}
    Z_n(\lambda_n)=0, \quad \text{where} \quad Z_n(\lambda)=\int_{\R^2} \left(\frac3{5 c_{TF}}(\lambda-\Phi_{n})_+\right)^{3/2}- Z.
\end{equation}
Note that, since $Z_n$ is a non-decreasing function of $\lambda$, a simple dichotomy can be used to compute efficiently $\lambda_n$. We then update the density by setting
\begin{equation}
    \widetilde{\rho_{n+1}}=\left(\frac3{5 c_{TF}}(\lambda_n-\Phi_n)_+\right)^{3/2} \quad \text{and} \quad \rho_{n+1}=t\widetilde{\rho_{n+1}}+(1-t)\rho_n,
\end{equation}
where $t\in[0,1]$ is optimized to minimize the Thomas-Fermi energy (we use again the Python predefined function \texttt{minimize}). The algorithm is terminated when a tolerance $\varepsilon$ is reached between two successively evaluated energies. We used $\varepsilon=10^{-5}$. 

For the approximation of the regularized energy $\mathcal{D}(\rho-\mu)$ by trapezoid rule, we use the formula~\eqref{energy_by_potential} involving the potential, which is valid since $\rho-\mu$ is neutral, instead of the one involving the Fourier transform. 


\subsection{Numerical results}
We present here the numerical results obtained for our test case, which corresponds to a homogeneous charged nanowire localized in a given square
\begin{equation}\label{eq:mu1}
    \mu_1(x)=\mathds{1}(|x_1|<2)\mathds{1}(|x_2|<2).
\end{equation}
We note that we take $c_{TF}=\frac23 \pi^2$. We start by plotting $\mu_1$ in Figure~\ref{fig:mu1}. Then, Figures~\ref{fig:resultsmu1} and~\ref{fig:resultsregmu1} show the results for the non-regularized and regularized models, respectively. The two models' combined results are fairly close. Figure~\ref{fig:resultserrormu1} shows the error between the two models.
We are able to say the following. The difference between the two densities appears to be almost zero far from the nanowire. The relative error is of order $13\%$ between potentials and of order $11\%$ between densities in the material's vicinity. Ultimately, there is a $7\%$ error between the two ground states energies in the two models.

\begin{figure}[H]
\centering
\begin{subfigure}{0.45\textwidth}
    \includegraphics[width=\textwidth]{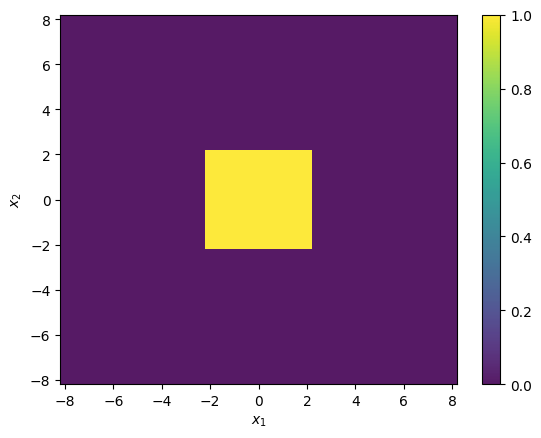}
    \caption{2D representation.}
    \label{fig:first}
\end{subfigure}
\hfill
\begin{subfigure}{0.45\textwidth}
    \includegraphics[width=\textwidth]{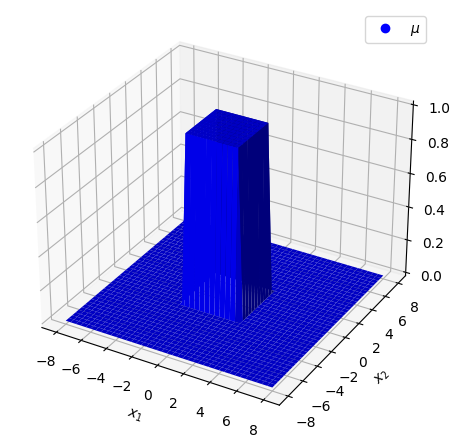}
    \caption{3D representation.}
    \label{fig:second}
\end{subfigure}
        
\caption{Nuclear density $\mu_1$ of a simple nanowire.}
\label{fig:mu1}
\end{figure}


\begin{figure}[H]
\centering
\begin{subfigure}{0.42\textwidth}
    \includegraphics[width=\textwidth]{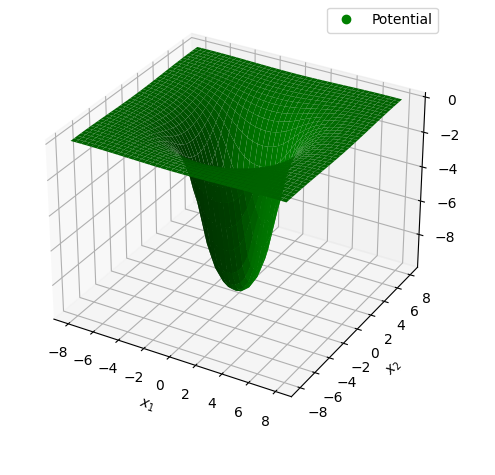}
    \\
    \includegraphics[width=\textwidth]{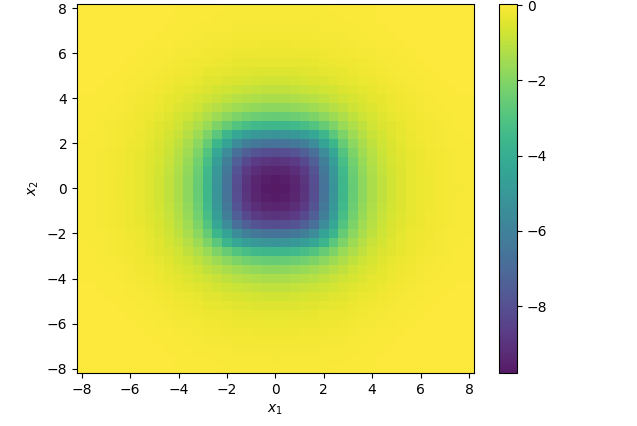}
    \caption{Potential $\tilde{\rV}$}
    \label{fig:pot3d}
\end{subfigure}
\hfill
\begin{subfigure}{0.42\textwidth}
    \includegraphics[width=\textwidth]{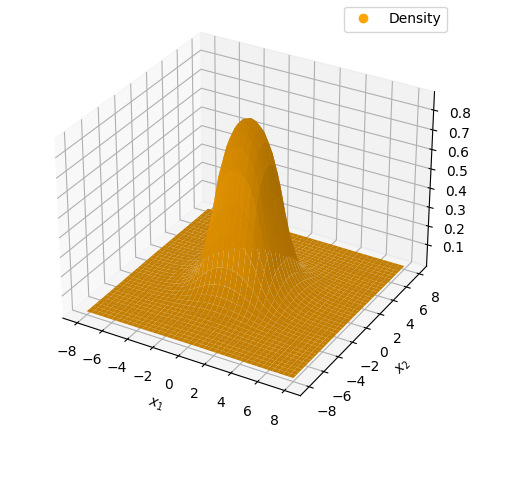}
    \\
    \includegraphics[width=\textwidth]{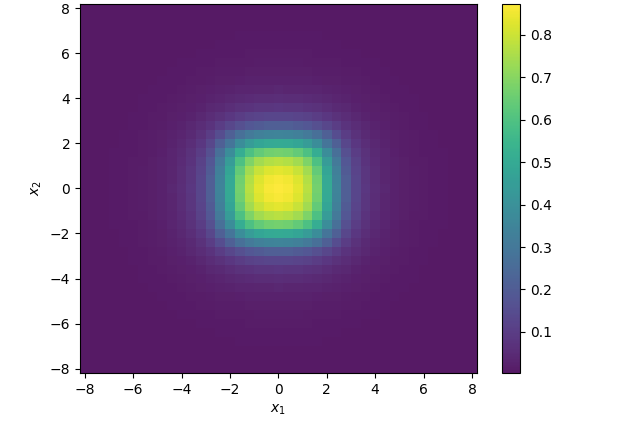}
    \caption{Density $\tilde{\rho}$}
    \label{fig:den3d}
\end{subfigure}
\caption{Results of the non-regularized model for $\mu_1$.}
\label{fig:resultsmu1}
\end{figure}

\vspace{-0.55cm}

\begin{figure}[H]
\centering
\begin{subfigure}{0.42\textwidth}
    \includegraphics[width=\textwidth]{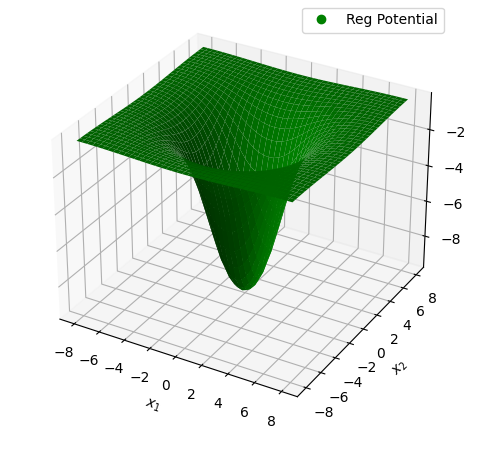}
    \\
    \includegraphics[width=\textwidth]{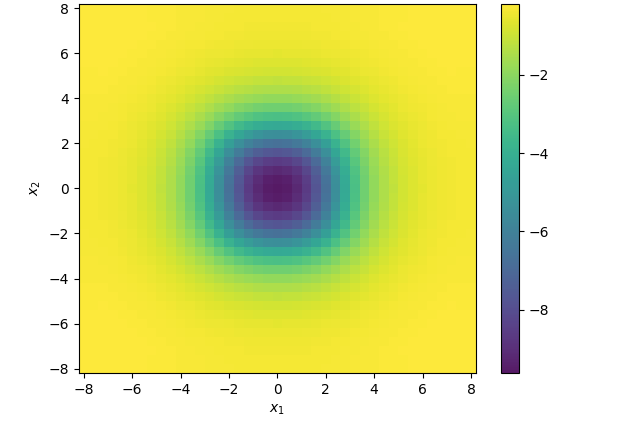}
    \caption{Regularized potential $\rV_{\TF}-\lambda$}
    \label{fig:pot3d}
\end{subfigure}
\hfill
\begin{subfigure}{0.42\textwidth}
    \includegraphics[width=\textwidth]{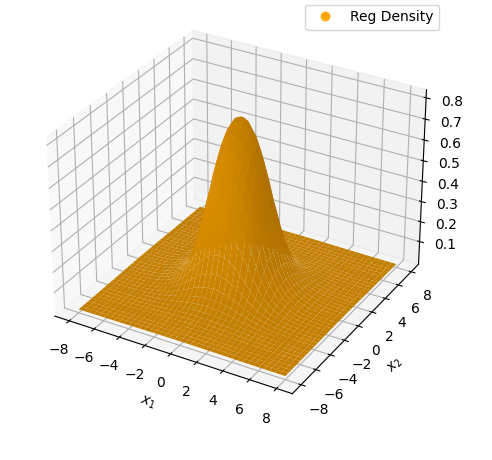}
    \\
    \includegraphics[width=\textwidth]{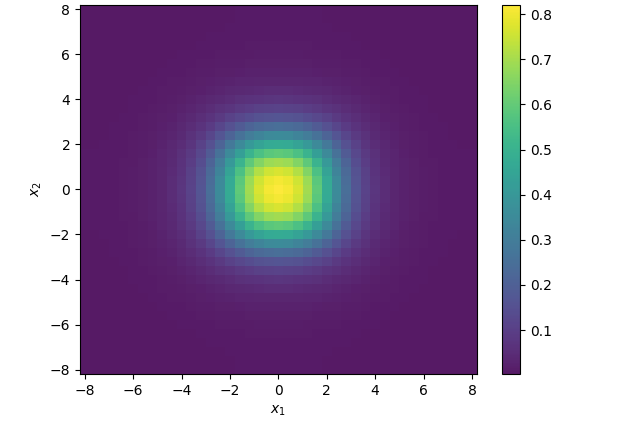}
    \caption{Regularized density $\rho_{\TF}$}
    \label{fig:den3d}
\end{subfigure}
\caption{Results of the regularized model for $\mu_1$.}
\label{fig:resultsregmu1}
\end{figure}

\begin{figure}[H]
\centering
\begin{subfigure}{0.45\textwidth}
    \includegraphics[width=\textwidth]{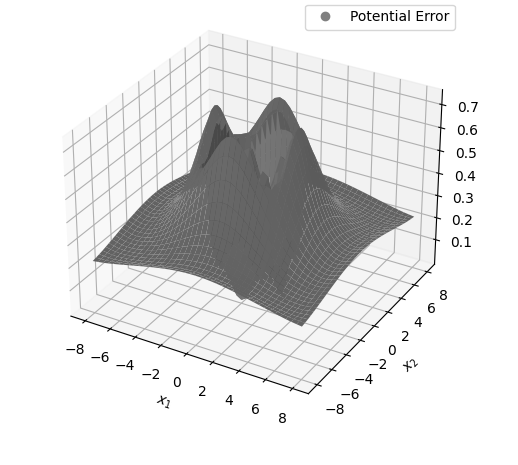}
    \\
    \includegraphics[width=\textwidth]{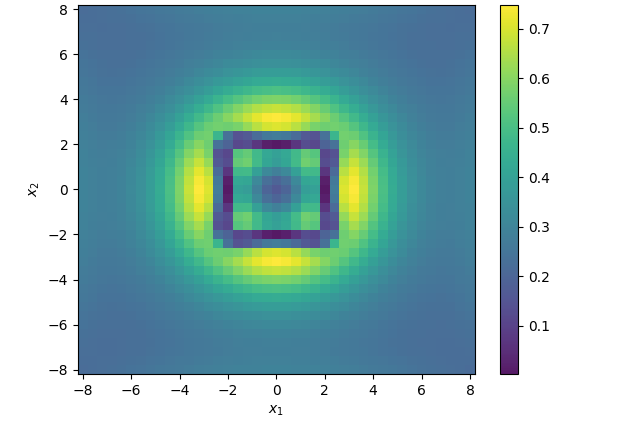}
    \caption{$|\rV_{TF}-\lambda-\tilde{\rV}|$}
    \label{fig:pot3d}
\end{subfigure}
\hfill
\begin{subfigure}{0.45\textwidth}
    \includegraphics[width=\textwidth]{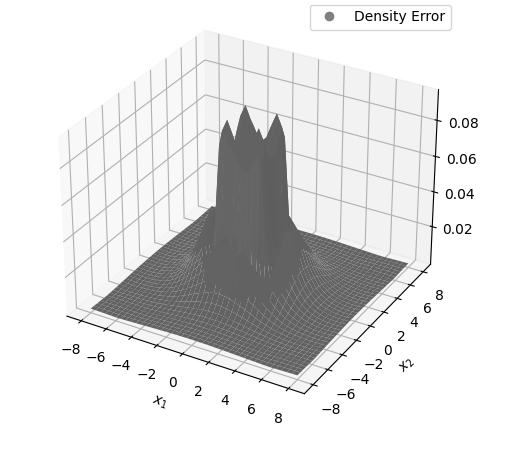}
    \\
    \includegraphics[width=\textwidth]{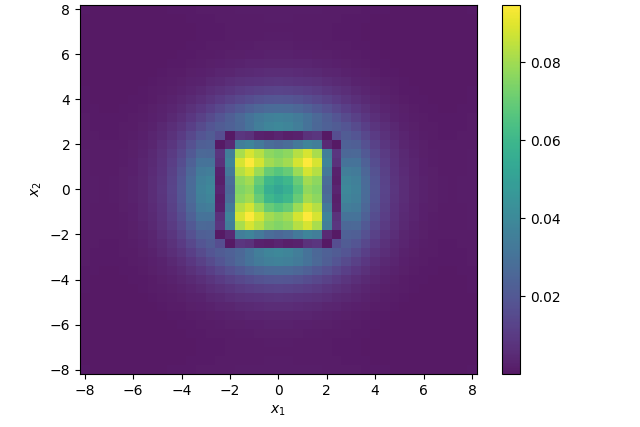}
    \caption{$|\rho_{TF}-\tilde{\rho}|$}
    \label{fig:den3d}
\end{subfigure}
\caption{Errors between the regularized and non-regularized models for $\mu_1$. }
\label{fig:resultserrormu1}
\end{figure}

From a computational view point, the regularized model ends in 1902s with 13 iterations versus 2025s for the non-regularized model. Moreover, the energy obtained by the regularized model is lower than that of the non-regularized model (83,11 versus 89,46). More in-depth study remains to be done with finer grids and other test cases, but this first simple test allows to see the interest of the regularized model with respect to the non-regularized model, where there are no theoretical results for the existence and uniqueness of the minimizer.

\bibliography{ref}
\bibliographystyle{siam}

\end{document}